\documentclass[letterpaper,journal]{IEEEtran}
\usepackage{amssymb,amsmath}
\usepackage{amsfonts}
\usepackage{graphicx}
\usepackage{psfrag}
\usepackage{color}
\usepackage{revsymb}
\usepackage{rotating}
\usepackage[all]{xy} 

\psfrag{Crate}{$r$}
\psfrag{Qrate}{$S$}
\psfrag{R}{$R$}
\psfrag{S}{$S$}

\DeclareMathOperator{\Tr}{Tr}

\DeclareMathOperator{\E}{\mathbb{E}}

\begin{document}
\def\isom{\simeq}
\def\norm#1{ {|\hspace{-.022in}|#1|\hspace{-.022in}|} }
\def\Norm#1{ {\big|\hspace{-.022in}\big| #1 \big|\hspace{-.022in}\big|} }
\def\NOrm#1{ {\Big|\hspace{-.022in}\Big| #1 \Big|\hspace{-.022in}\Big|} }
\def\NORM#1{ {\left|\hspace{-.022in}\left| #1 \right|\hspace{-.022in}\right|} }

\def\pmat#1{\begin{pmatrix} #1 \end{pmatrix}}
\def\qedsymbol{\rule{7pt}{7pt}}
\def\supp{ {\rm{supp \,}}}
\def\dist{ {\rm{dist }}}
\def\dim{ {\rm{dim \,}}}
\def\oti{{\otimes}}
\def\bra#1{{\langle #1 |  }}
\def\lb{ \left[ }
\def\rb{ \right]  }
\def\tilde{\widetilde}
\def\bar{\overline}
\def\*{\star}
\def\({\left(}		\def\BL{\Bigr(}
\def\){\right)}		\def\BR{\Bigr)}
	\def\BBL{\lb}
	\def\BBR{\rb}


\def\1{{\mathbf{1} }}

\def\bb{{\bar{b} }}
\def\ab{{\bar{a} }}
\def\zb{{\bar{z} }}
\def\zbar{{\bar{z} }}
\def\inv#1{{1 \over #1}}
\def\half{{1 \over 2}}
\def\d{\partial}
\def\der#1{{\partial \over \partial #1}}
\def\dd#1#2{{\partial #1 \over \partial #2}}
\def\vev#1{\langle #1 \rangle}
\def\ket#1{ | #1 \rangle}
\def\rvac{\hbox{$\vert 0\rangle$}}
\def\lvac{\hbox{$\langle 0 \vert $}}
\def\2pi{\hbox{$2\pi i$}}
\def\e#1{{\rm e}^{^{\textstyle #1}}}
\def\grad#1{\,\nabla\!_{{#1}}\,}
\def\dsl{\raise.15ex\hbox{/}\kern-.57em\partial}
\def\Dsl{\,\raise.15ex\hbox{/}\mkern-.13.5mu D}
\def\b#1{\mathbf{#1}}
\newcommand{\proj}[1]{\ket{#1}\bra{#1}}
\def\braket#1#2{\langle #1 | #2 \rangle}
\def\1s2#1{\frac{1}{\sqrt{2^{#1}}}}
%
%
\def\th{\theta}		\def\Th{\Theta}
\def\ga{\gamma}		\def\Ga{\Gamma}
\def\be{\beta}
\def\al{\alpha}
\def\ep{\epsilon}
\def\vep{\varepsilon}
\def\la{\lambda}	\def\La{\Lambda}
\def\de{\delta}		\def\De{\Delta}
\def\om{\omega}		\def\Om{\Omega}
\def\sig{\sigma}	\def\Sig{\Sigma}
\def\ph{\varphi}
\def\d{\delta}
                        \def\Up{\Upsilon}
%
%
\def\CA{{\cal A}}	\def\CB{{\cal B}}	\def\CC{{\cal C}}
\def\CD{{\cal D}}	\def\CE{{\cal E}}	\def\CF{{\cal F}}
\def\CG{{\cal G}}	\def\CH{{\cal H}}	\def\CI{{\cal J}}
\def\CJ{{\cal J}}	\def\CK{{\cal K}}	\def\CL{{\cal L}}

\def\CM{{\cal M}}	\def\CN{{\cal N}}	\def\CO{{\cal O}}
\def\CP{{\cal P}}	\def\CQ{{\cal Q}}	\def\CR{{\cal R}}
\def\CS{{\cal S}}	\def\CT{{\cal T}}	\def\CU{{\cal U}}
\def\CV{{\cal V}}	\def\CW{{\cal W}}	\def\CX{{\cal X}}
\def\CY{{\cal Y}}	\def\CZ{{\cal Z}}

\def\rvac{\hbox{$\vert 0\rangle$}}
\def\lvac{\hbox{$\langle 0 \vert $}}
\def\comm#1#2{ \BBL\ #1\ ,\ #2 \BBR }
\def\2pi{\hbox{$2\pi i$}}
\def\e#1{{\rm e}^{^{\textstyle #1}}}
\def\grad#1{\,\nabla\!_{{#1}}\,}
\def\dsl{\raise.15ex\hbox{/}\kern-.57em\partial}
\def\Dsl{\,\raise.15ex\hbox{/}\mkern-.13.5mu D}
\def\beq{\begin {equation}}
\def\eeq{\end {equation}}
\def\to{\rightarrow}
\def\h#1{\widehat{#1}}
\def\inn{{\rm in}}
\def\out{{\rm out}}
\def\sim{{\rm sim}}
\def\ave{{\rm ave}}
\def\br#1{\langle #1 \rangle}
\def\ox{\otimes}
\newtheorem{lem}{Lemma}
\newtheorem{prop}{Proposition}
\newtheorem{theo}{Theorem}
\newtheorem{cor}{Corollary}

\newtheorem{rem}{Remark}
\newtheorem{dfn}{Definition}

\def\12{{\textstyle \frac 12}}
\renewenvironment{proof}[1][Proof]{\emph{{#1~}\!\!:}}{\QED}

\def\bC{\mathbb{C}}
\def\diag{\mbox{diag}}
\def\nn{\nonumber}
\def\bs#1{\boldsymbol{#1}}
\def\id{\text{\rm id}}

\def\argmax{\mbox{argmax}}
\definecolor{gray}{gray}{.9}
\def\com#1{\vspace{.1in}\fcolorbox{black}{gray}{\begin{minipage}{5.5in}#1\end{minipage}}\vspace{.1in}}

\def\deph{\Delta}
\def\lmin{\lambda_{\text{min}}}
\def\lmax{\lambda_{\text{max}}}
\def\abs{{\rm abs}}

\def\sumk{ {\frac{1}{\kappa}\sum_{k=1}^\kappa}}

\title{Optimal quantum source coding with quantum side information at the encoder and decoder}
\author{
Jon Yard$^*$, 
 \thanks{$*$  {\tt jtyard@lanl.gov}, 
Institute for Quantum Information, California Institute of Technology, Pasadena, California, USA, CNLS (Quantum Initiative), CCS-3, Los Alamos National Laboratory, Los Alamos, NM, USA}
Igor Devetak$^\dagger$ 
 \thanks{$\dagger$  {\tt devetak@usc.edu},
 Electrical Engineering Department, University of Southern California, USA}
}
\date{May 22, 2007}
\maketitle
\begin{abstract} 
Consider many instances of an arbitrary quadripartite pure state of four quantum systems $ABCD$.  Alice holds the $AC$ part of each state, Bob holds $B$, while $D$ represents all other parties correlated with $ABC$.  Alice is required to redistribute the $C$ systems to Bob while asymptotically preserving the overall purity.  We prove that this is possible using $Q$ qubits of communication and $E$ ebits of shared entanglement between Alice and Bob, provided that
$Q~\!\geq~\!\12I(C;D|B)$~and~$Q~\!+~\!E~\!\geq~\!H(C|B),$
proving the optimality of the Luo-Devetak outer bound.  The optimal qubit rate provides the first known operational interpretation of quantum conditional mutual information.  We also show how our protocol leads to a fully operational proof of strong subadditivity and uncover a general organizing principle, in analogy to thermodynamics, that underlies the optimal rates.  
\end{abstract}


\section{Introduction}
\PARstart{T}{he} most fundamental problem in communication theory is the two-terminal source coding problem.  Here one user, say Alice, attempts to describe a source of information to another user, who we call Bob.  If the information source is modeled by a sequence of independent and identically distributed (i.i.d.) random variables $X$, one can ask for the ultimate rate at which the source can be described, in units of bits per sample.  It is required that Alice's description allow Bob to perfectly recreate the source sequence with high probability, although decreasing the error probability generally requires block coding on longer source sequences.   
According to Shannon's noiseless channel coding theorem \cite{shannon}, this ultimate rate is given by the \emph{Shannon entropy}
\[H(X) = -\sum_x p(x)\log p(x).\]
Intuitively, Shannon entropy can be understood as a measure of the information contained in the random variable $X$.  Because Shannon entropy answers the question regarding the optimal rate for data compression, one says that the corresponding protocol for data compression provides an \emph{operational interpretation} of Shannon entropy.

Suppose now that Bob had some \emph{a priori} information about $X$, in the form of a correlated random variable $Y$.  In this case, Slepian and Wolf demonstrated \cite{SW71} that Alice would only need to send to Bob at a rate given by the \emph{conditional entropy}
\[H(X|Y) = H(XY) - H(Y)\]
and that surprisingly, Alice would not need to know Bob's side information to accomplish this task.  The so-called Slepian-Wolf protocol for data compression with side information provides an operational interpretation of conditional entropy.  Intuitively, one thinks of $H(X|Y)$ as a measure of the information that is to be gained by learning $X$ for one who already knows $Y$.  Note that there is no advantage if Alice has additional side information regarding $X$, and that shared common randomness between Alice and Bob is also of no help.  

In this paper, we provide a complete solution to a general quantum counterpart of the above scenario.  We find that, in contrast to the classical case, additional Alice side information changes the problem, while quantum mechanical entanglement between Alice and Bob, the quantum analog of shared common randomness, is a useful resource.  Our problem is \emph{fully quantum} in a sense introduced by Schumacher \cite{Sch95}, where Alice is asked to transfer part of a pure quantum state to Bob, while preserving the purity of the global state.    
For this, we consider a pure state of four quantum systems $\ket{\psi}^{ABCD}$.   Initially, the $A$ and $C$ systems are held by Alice, while $B$ is in the possession of Bob.  We refer to $D$ as the \emph{reference system} and assume that it is inaccessible to both Alice and Bob.  We determine the cost for Alice and Bob to ``redistribute" the state, so that it is Bob who holds $C$ instead of Alice, thereby transferring the quantum information in $C$ to Bob.  Specifically, we analyze the corresponding asymptotic scenario, asking that many copies of the same state be redistributed as above, while requiring that the redistributed states have arbitrarily high fidelity with the originals in the asymptotic limit.

To achieve this task, we allow the use of two fundamental quantum mechanical resources.  First, Alice may send \emph{qubits} (two-level quantum systems) to Bob over a noiseless quantum channel.  Second, we allow Alice and Bob to use pre-existing entanglement, shared between themselves in the form of Bell states 
\[\ket{\Phi^+} = \frac{1}{\sqrt{2}}\big(\ket{00} + \ket{11}\big).\]
We refer to such a state as an \emph{ebit} (entangled bit).
We do not separately consider classical communication, because it can be used with entanglement to simulate qubit channels via teleportation.
The asymptotic cost to redistribute $C$ as above is given in terms of the number $Q$ of qubits sent and the number $E$ of ebits consumed, per copy of the state.  We allow the entanglement cost $E$ to be negative, in which case the corresponding protocol generates entanglement rather than consume it.  Our main result (Theorem~\ref{theo:main}) proves the optimality of the Luo-Devetak outer bound \cite{LD06} for this problem, demonstrating that it is possible to redistribute the state $\ket{\psi}^{ABCD}$ as above if and only if 
\begin{IEEEeqnarray}{rClrCl}
Q &\geq& \12 I(C;D|B), & \hspace{.3in}
Q + E &\geq& H(C|B).
\label{eqn:mainregion}
\end{IEEEeqnarray}
This region is depicted in Figure~\ref{fig:mainregion}.
The quantities in these bounds,  \emph{conditional mutual information} and \emph{conditional entropy}, are defined in Section~\ref{section:notation}.   Simultaneously minimizing the qubit rate $Q$ and the total sum rate $Q + E$ gives the \emph{optimal cost pair}
\begin{IEEEeqnarray}{rClrCl}
Q^* &=& \12 I(C;D|B) & \hspace{.2in} E^* &=& \12 I(A;C) - \12I(B;C).
\label{eqn:corner}
\end{IEEEeqnarray}
The optimal qubit cost gives the first known operational interpretation of quantum conditional mutual information.  In Section~\ref{section:ssad}, we show that $Q^*$ cannot be negative, which leads to an operational proof of the celebrated strong subadditivity inequality \cite{LR73}.  This proof differs from other such operational proofs \cite{GPW05,HOW05}  in that it follows solely from a direct coding theorem and not from a converse proof. 
In \cite{DY06b}, where our main result was first announced, we showed that $Q^*$ is symmetric under time-reversal, where now Bob redistributes $C$ back to Alice, while $E^*$ is anti-symmetric.  The former gives an intuitive understanding to the curious identity 
\[I(C;D|A) = I(C;D|B),\]
which holds on every quadripartite pure state.  
We comment further on this feature in Section~\ref{section:discussion}.
We also demonstrated there that the corresponding protocol is \emph{perfectly composable}. This constitutes an exact solution to a quantum analog of result of Cover and Equitz \cite{CE91} on the successive refinement of classical information, although the classical problem is only known to be exactly soluble in the presence of a Markov condition. 

\begin{figure}
\begin{centering}
\hspace{.2in} \includegraphics[scale=.5]{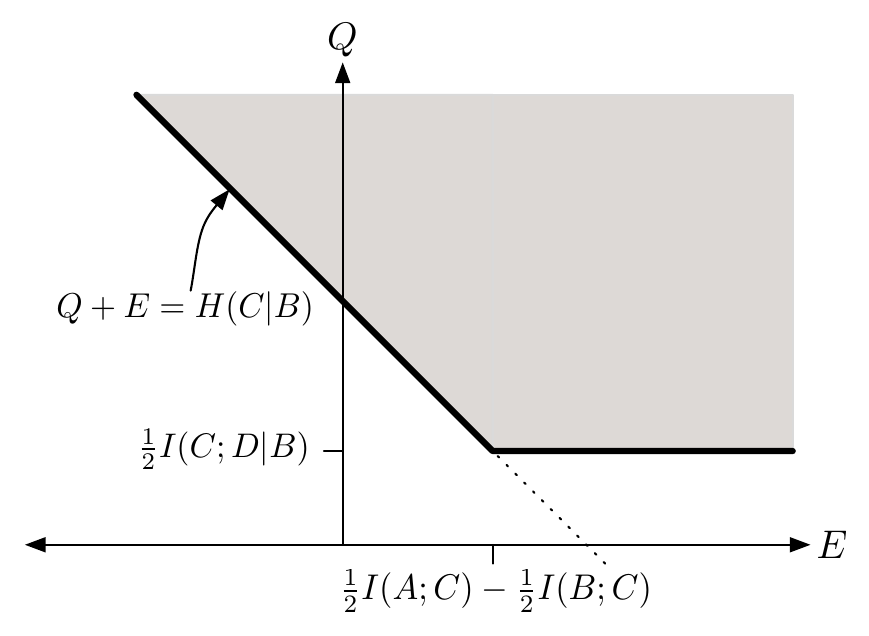}
\end{centering}
\caption{The shaded region represents contains the cost pairs from (\ref{eqn:mainregion}) at which it is possible to  redistribute the $C$ part of $\ket{\psi}^{ABCD}$ from Alice to Bob.  The figure corresponds to the case $I(A;C) > I(B;C)$; otherwise the corner point,  which corresponds to the optimal cost pair in (\ref{eqn:corner}), would be in the upper-left quadrant.}
\label{fig:mainregion}
\end{figure}

By assuming that various subsystems are trivial, the state redistribution problem generalizes numerous tasks that were previously considered in the literature while giving an optimal protocol suited for any and all of them.    
As we discuss in Section~\ref{section:discussion} (in particular see Figure~\ref{figure:specialcases}) and also during the proof of our main theorem in Section~\ref{sec:fqrsproof}, these tasks include Schumacher compression \cite{Sch95}, state merging and splitting \cite{HOW05,HOW05b, D05b, ADHW06}, and entanglement concentration and dilution \cite{BBPS96}.  We depart from previous nomenclature with regard to the merging and splitting problems; our convention for this paper is detailed in Section~\ref{sec:fqrsproof}. 

The paper is organized as follows.  In the next subsection we fix our notational conventions.  The following section gives an introduction to the resource calculus.  There we also formally state the main result, Theorem~\ref{theo:main}, which is proved in Section~\ref{section:mainproof}.  In Section~\ref{section:ssad},
we show how our results yield a fully operational proof of strong subaddivity which, unlike previous operational proofs, is logically independent even from the subadditivity of entropy.
We conclude with a discussion in Section~\ref{section:discussion} where we reflect on the main result and provide a novel thermodynamic interpretation of the optimal rates.      


\subsection{Notational conventions \label{section:notation}}
Throughout this paper, we assume familiarity with standard background material in quantum information theory; for a general reference, the reader is referred to \cite{NC00a}.
We use capital Roman letters such as $A,B,C$ to denote Hilbert spaces.  We write $|A|$ for the dimension of $A$ and use a superscripted label to associate a state to a Hilbert space, by writing $\rho^A$ or $\ket{\ph}^A$.   Computational basis states of $A$ are denoted with lower case Roman letters as in $\{\ket{i}^A\}$.
Tensor products of Hilbert spaces are written $AB = A\ox B$.  Given a pure state $\ket{\ph}^{AB}$, we abbreviate $\ph^{AB} =  \proj{\ph}^{AB}$, while writing its partial traces as $\ph^A = \Tr_B \ph^{AB}$.  We write $\pi^A$ for the maximally mixed state on $A$, and given two isomorphic Hilbert spaces $A$ and $A'$, we write 
\[\ket{\Phi}^{AA'} = \frac{1}{\sqrt{|A|}}\sum_{i=1}^{|A|} \ket{i}^A\ket{i}^{A'}\]
for the unique maximally entangled state associated with the isomorphism $\ket{i}^{A} \mapsto \ket{i}^{A'}$.  
A \emph{quantum channel} is a completely positive, trace-preserving linear map $\CN^{A\to B}$ from density matrices on $A$ to those on $B$.  Given an isometry $\CV^{A\to B}$, we will abbreviate its adjoint action on density matrices as $\CV(\rho) \equiv \CV \rho \CV^\dagger$.  A \emph{partial isometry} is an isometry when restricted to its support subspace.  

For the von Neumann entropy of a density matrix $\ph^A$ we write 
\[H(A) \equiv -\Tr\ph^A\log_2\ph^A.\]
When the underlying state could be ambiguous we write $H(A)_\ph$.  Given a multipartite state $\ph^{ABC}$, various entropic quantities can be defined in exact analogy to the classical case (see e.g.\ \cite{CT91a}).  \emph{Quantum conditional entropy} is defined \cite{CA97a} as 
\[H(A|B) = H(AB) - H(B),\]
\emph{quantum mutual information} \cite{CA97a} is 
\[I(A;B) = H(A) + H(B) - H(AB)\]
and \emph{quantum conditional mutual information} is given by 
\[I(A;B|C) = H(A|C) + H(B|C) - H(AB|C).\]
Observe that the conditional quantities above cannot generally be interpreted as averages, unless the conditioning system is purely classical.   Furthermore, notice that conditional entropy can in fact be negative, as it is for any pure entangled state on $AB$.  On the other hand, $I(A;B|C)$ is never negative, a fact that is known as \emph{strong subaddivity} \cite{LR73}.  In Section~\ref{section:ssad}, we show how our main result leads to a self-contained proof of strong subadditivity.  


\section{Resource inequalities}
It will be convenient for us to use the high-level notation of  \emph{resource inequalities} \cite{DHW04,DHW05} to express our main result, as well as to describe various intermediate protocols introduced during the proof.  We use a more elementary formulation than \cite{DHW05} which is nonetheless sufficient for our purposes. 
\subsection{Finite resource inequalities}
A single ebit shared between Alice and Bob is denoted $[qq]$.  The notation $[q\!\to \!q]$ represents a noiseless qubit channel from Alice to Bob, while a noiseless classical bit channel is written $[c\to c]$.  A \emph{finite resource inequality} is an expression such as 
\[[q\to q] \geq [c\to c], \hspace{.2in} [q\to q] \geq [qq]\]
meaning that the resource on the left can simulate the one on the right.  The above two examples respectively signify that a qubit channel can be used to send classical bits (by signaling with orthogonal pure states), or otherwise can be used to distribute entanglement (by transmitting halves of locally prepared ebits).    
Addition of two resources may be regarded as having each of them available.  In this way, for instance, the existence of the quantum teleportation and superdense coding protocols are proofs of the respective finite resource inequalities 
\begin{IEEEeqnarray}{rCl}
[qq] + 2[c\to c] \geq [q\to q], \hspace{.1in} [qq] + [q\to q] \geq 2[c\to c].\hspace{.1in} \label{eqn:RIexamples}
\end{IEEEeqnarray}


\subsection{Approximate resource inequalities}
Given two quantum states $\rho$ and $\sig$ of the same quantum system, we may judge their closeness using either the trace distance $\norm{\rho - \sig}_1$ or the \emph{fidelity}  $F(\rho,\sig) = \Norm{\sqrt{\rho}\sqrt{\sig}}_1^2$.  Note that when one of the states is pure, $F(\ket{\ph},\sig) = \bra{\ph}\sig\ket{\ph}$.
 A useful characterization of fidelity -- \emph{Uhlmann's theorem} -- says that if $\ket{\psi}$ is a purification of $\rho$, then $F(\rho,\sig)$ is the maximum of $|\braket{\psi}{\phi}|^2$ over all purifications $\ket{\phi}$ of $\sig$.   Fidelity and trace distance related by the inequalities 
\begin{IEEEeqnarray}{rCl}
F(\rho,\sig) &\geq& 1 - \norm{\rho - \sig}_1 \label{eqn:tr2fid} \\
\norm{\rho - \sig}_1 &\leq& 2\sqrt{1-F(\rho,\sig)}.  \label{eqn:fid2tr}
\end{IEEEeqnarray}
Therefore, fidelity and trace distance are equivalent distance measures when one is interested in arbitrarily good approximations of states as we are here.  
An \emph{approximate resource inequality} 
\[\sum_i a_i \geq_\ep \sum_j b_j\]
is a finite resource inequality that holds with an error of $\ep$ in the following sense.  Consider acting on half of a maximally entangled state with each target resource $b_j$ that is a channel, and call the resulting global state $\Omega$.  Note that $\Omega$ should also contain the $b_j$ that are quantum states.  Now, let $\Omega'$ be the simulated version of this state, obtained by using the resources $a_i$.  We require that $\Om$ and $\Om'$ are $\ep$-close in either trace distance or fidelity.  The particular measure is not important, as we are ultimately concerned with asymptotics, where $\ep$ can be arbitrarily small.  


\subsection{Asymptotic resource inequalities}
The notion of a finite resource inequality can be generalized to that of an \emph{asymptotic resource inequality}.  This is a formal expression of the form 
\begin{IEEEeqnarray}{rCl}
\sum_i R_\inn^{(i)}a_i \succeq \sum_j R_\out^{(j)}b_j. \label{eqn:RI}
\end{IEEEeqnarray}
Here the $a_i$ and $b_j$ are resources and the rates $R_\inn^{(i)}$ and $R_\out^{(j)}$ are nonnegative real numbers.  We shall consider the inequality (\ref{eqn:RI}) to be shorthand for the following formal statement:  for every $\ep > 0$, every set of rates $R'^{(i)}_\inn > R^{(i)}_\inn$, $R'^{(i)}_\out < R^{(i)}_\out$ and all sufficiently large $n$, the approximate resource inequality 
\begin{IEEEeqnarray*}{rCl}
\sum_i \lfloor nR_\inn'^{(i)}\rfloor a_i \geq_\ep \sum_j \lfloor nR_\out'^{(j)}\rfloor b_j 
\end{IEEEeqnarray*}
holds.  Below, we use Greek letters to denote linear combinations of finite resources that appear in asymptotic resource inequalities.  In some asymptotic resource inequalities, we may only require a sublinear amount $o(n)$ of a particular input resource.  In such cases, we write $o a + \beta \succeq \gamma $ if we have $Ra + \beta\succeq \gamma$ for every $R > 0$. 

It will also be convenient for us extend the definition of asymptotic resource inequalities to have negative rates on the left.  Such rates are interpreted as meaning that the corresponding resources are {generated} rather than consumed.  Formally, these resources should be negated and moved to the right.   Let us introduce two powerful lemmas that are the raison d'\^{e}tre for the entire formalism of asymptotic resource inequalities and which play important roles in our proofs.
\begin{lem}[Composition lemma \cite{DHW05}]
\[\al \succeq \beta \text{ and } \beta \succeq \gamma \Rightarrow \al \succeq \gamma.\]
\label{lem:composition}
\end{lem}
\begin{lem}[Cancellation lemma \cite{DHW05}] 
Given rates that satisfy  $R_\inn > R_\out \geq 0$,  
\[R_\inn a + \beta \succeq R_\out a + \gamma \Rightarrow (R_\inn -R_\out)a + \beta \succeq \gamma.\]
Otherwise, if $R_\out \geq R_\inn \geq 0$, then 
\[R_\inn a + \beta \succeq R_\out a + \gamma \Rightarrow  o a + \beta\succeq (R_\out -R_\inn)a +\gamma\]
\label{lem:cancellation}
\end{lem}


\subsection{Distributed states}
In Schumacher data compression, Alice wishes to transmit the $C$ parts of many instances of the state $\ket{\psi}^{CD}$ to Bob while asymptotically preserving the entanglement with $D$.  We introduce the following notation to describe the corresponding coding theorem:
\[\psi^{C|\emptyset} + H(C) [q\to q] \succeq \psi^{\emptyset|C}.\]
The notation $\psi^{C|\emptyset}$ indicates that Alice holds the $C$ parts of many i.i.d.\ instances of some fixed purification $\ket{\psi}^{CD}$ of the density matrix $\psi^C$, while Bob holds nothing.  On the right, the expression $\psi^{\emptyset|C}$ refers to the same purifications as on the left, only it is Bob who is holding the $C$ systems.  In other words, Alice attempts to simulate identity channels from the systems $C$ in her lab to identical systems $C$ located in Bob's lab.  This channel is only required to work well when the input is equal to $\psi^C$.  In \cite{DHW05}, the formalism of \emph{relative resources} was introduced for these purposes, though our alternate notation is sufficient for our needs.   
State redistribution involves a purification $\ket{\psi}^{ABCD}$ of a tripartite density matrix $\psi^{ABC}$.  We denote the distributed states before and after the protocol as $\psi^{AC|B}$ and $\psi^{A|CB}$  since Alice begins by holding $AC$ and ends by only holding $A$.  The rates in an asymptotic resource inequality involving such distributed states will in general be entropic expressions evaluated on the implicit but arbitrary purification into a reference system $D$.
Using this notation, we again state the main result:
\begin{theo}
\begin{IEEEeqnarray}{rCl}
\psi^{AC|B} + Q [q\to q] +   E [qq] \succeq \psi^{A|CB}
\label{eqn:main}
\end{IEEEeqnarray}
 if and only if $Q$ and $E$ satisfy (\ref{eqn:mainregion}), i.e.\ are contained in the region depicted in Figure~\ref{fig:mainregion}. 
\label{theo:main}
\end{theo}
The converse part of the proof of Theorem~\ref{theo:main}, i.e.\ that $Q$ and $E$ must satisfy (\ref{eqn:mainregion}), is proved in \cite{LD06}. We thus focus on proving a coding theorem showing that (\ref{eqn:main}) is satisfied whenever $Q$ and $E$ satisfy (\ref{eqn:mainregion}).  Because $[q\to q] \geq [qq]$, it suffices for us   
to demonstrate (\ref{eqn:main}) for the corner point 
 $(Q^*, E^*)$ defined in (\ref{eqn:corner}).


\section{Proof of Theorem~\ref{theo:main} \label{section:mainproof}}

To prove Theorem~\ref{theo:main} we will demonstrate the existence of the following auxiliary protocol that transfers $C^n$ to Bob and has the desired net communication and entanglement cost:    
\begin{theo}
\begin{IEEEeqnarray*}{L}
\psi^{AC|B} + \12 I(C;BD)[q\to q] + \12 I(A;C) [qq] \hspace{2.5in}\\
\hspace{.8in} \succeq  \psi^{A|CB} + \12I(B;C)[q\to q] + \12I(B;C) [qq].
\end{IEEEeqnarray*}
\label{theo:aux} 
\end{theo}
Together with the cancellation lemma (Lemma~\ref{lem:cancellation}), Theorem~\ref{theo:aux} yields a proof of Theorem~\ref{theo:main}.  However, observe that if $I(B;C) \geq I(A;C)$, the cancellation lemma still requires a sublinear amount of entanglement on the left.  Similarly, if we have $I(C;BD) = I(C;D)$ (i.e.\ if strong subaddivity is saturated), a sublinear amount of communication will also be required.  However, because $[q\to q] \geq [qq]$, the additional entanglement cost can be absorbed into the communication rate and is therefore only relevant if the state $\psi^{CBD}$ saturates strong subaddivity.  We discuss this point further in Section~\ref{section:discussion}. 

We prove Theorem~\ref{theo:aux} by means of another protocol that simulates  \emph{coherent channels} \cite{H04}.  A coherent channel  $[q\to qq]$ is a type of quantum feedback channel that is an isometry from Alice to Alice and Bob:
\[\ket{0}^A\ket{0}^B\bra{0}^A + \ket{1}^A\ket{1}^B\bra{1}^A.\]
Using a coherent version of teleportation, where Alice and Bob apply only local unitaries, it is known that \cite{H04}
\[[qq] + 2[q\to qq] \geq 2[qq] + [q\to q].\]
Repeated concatenation yields the following asymptotic resource inequality \cite{H04}: 
\begin{IEEEeqnarray}{rCl}
2 [q\to qq] \succeq [q\to q] + [qq].
\label{eqn:cobitidentity}
\end{IEEEeqnarray}
In fact, the opposite direction holds as a finite resource inequality, but it will not be useful for us here.  
In this paper, we devote most of our efforts toward proving the following theorem which, when combined with (\ref{eqn:cobitidentity}) and the composition lemma (Lemma~\ref{lem:composition}), provides a proof of Theorem~\ref{theo:aux}.
\begin{theo}
\begin{eqnarray*}
\psi^{AC|B}+ \12 I(C;BD)[q\to q] + \12 I(A;C)[qq] \hspace{.9in}\\
\hspace{1in}\succeq  \psi^{A|CB} + I(B;C)[q\to qq].
\end{eqnarray*}
\label{theo:fqrspiggyback}
\end{theo}


\subsection{Proof of Theorem~\ref{theo:fqrspiggyback}}
Our proof of Theorem~\ref{theo:fqrspiggyback} relies on the following one-shot version.  We call this a ``robust" one-shot protocol because the error bound is robust to small perturbations in the underlying state (c.f.\ \cite{HW03}).  We delay the proof of this theorem until Section~\ref{sec:fqrsproof}.  
\begin{theo}[Robust one-shot redistribution protocol]
Let a pure state $\ket{\psi}^{ABCD}$ and a maximally entangled state $\ket{\Phi}^{\h{A}\h{B}}$ be given, where $|\h{A}| = |\h{B}|$ divides $|C|$. Suppose that $\ket{\ph}^{ABCD}$ and $\ket{\phi}^{ABCD}$ are states satisfying 
\[\max\left\{\Norm{\psi^{ABCD} - \ph^{ABCD}}_1,\Norm{\psi^{ABCD} - \phi^{ABCD}}_1\right\} \leq \ep\]
for some $\ep \leq ( 6 - 4\sqrt{2})^2 \approx .1177$. 
Then there exist a quantum system $S$ with $|S| = |C|/|\h{B}|$,  $\kappa$ encoding isometries $\CV_k^{\h{A}AC\to AS}$ and a decoding isometry $\CW^{SB\h{B}\to BCK}$ under which 
\begin{IEEEeqnarray}{rCl}
\frac{1}{\kappa}\sum_{k=1}^\kappa \bra{k}^{K}\bra{\psi}^{ABCD}\CW\CV_k\ket{\psi}^{ABCD}\ket{\Phi}^{\h{A}\h{B}}&\geq& 1-\eta
\label{eqn:1shoterror}
\end{IEEEeqnarray}
where $\eta$ is equal to 
\begin{IEEEeqnarray}{rCl}
6\sqrt{\ep} + 4\left(\!\frac{|C| \norm{\ph^{BD}}_0
\norm{\ph^{BCD}}_2^2}
{|S|^2}\right)^{1/4}
\!\!\!\!\!\!\!+ \frac{4 \kappa \norm{\phi^{BC}}_0 \norm{\phi^B}_\infty}{ |C|}. \hspace{.2in}
\label{eqn:eta}
\end{IEEEeqnarray}
\label{theo:1shotfqrspiggyback}
\end{theo}

\vspace{.1in}
Now we show how to apply Theorem~\ref{theo:1shotfqrspiggyback} to pure states of the form $\big(\ket{\psi}^{ABCD}\big)^{\ox n}$ to obtain a proof of Theorem~\ref{theo:fqrspiggyback}.  This is accomplished via the following theorem.  The direct part is proved in \cite{ADHW06}, while the converse part follows from standard arguments in classical information theory (see e.g.\ \cite{CT91a}).
\begin{theo}[Method of types]
Let a tripartite state $\ket{\psi}^{ABC}$ be given.
For every $\ep, \delta > 0$ and all sufficiently large $n$, there are projections 
$\Pi_\d^{A^n}$, $\Pi_\d^{B^n}$, and  $\Pi_\d^{C^n}$
such that for $T\in \{A,B,C\}$, 
\begin{IEEEeqnarray*}{rCl}
\Tr\Pi_\d^{T^n}(\psi^T)^{\ox n} &\geq& 1-\ep
\end{IEEEeqnarray*}
\begin{IEEEeqnarray*}{rCcCl}
2^{nH(T) - n\d} &\leq& \Tr\Pi_\d^{T^n} &\leq& 2^{nH(T) + n\d}.
\end{IEEEeqnarray*}
Also, the
normalized version $\ket{\ph}^{A^nB^nC^n}$ of the subnormalized state 
\[\big(\Pi_\d^{A^n}\!\ox\Pi_\d^{B^n}\!\ox\Pi_\d^{C^n}\big)\big(\ket{\psi}^{ABCD}\big)^{\ox n}\]
satisfies 
\begin{IEEEeqnarray*}{rCL}
\Norm{\ph^{A^n\!B^n\!C^n} - \big(\psi^{ABC}\big)^{\ox n}}_1 &\leq& \ep 
\end{IEEEeqnarray*}
and for each $T\in \{A,B,C,AB,BC,AC\}$, 
\begin{IEEEeqnarray*}{rClCl}
 2^{nH(T) - n\d} &\leq& \,\,\,\Norm{\ph^{T^n}}_0 &\leq& 2^{nH(T) + n\d} \\
 2^{-nH(T) - n\d} &\leq& \,\,\,\Norm{\ph^{T^n}}^2_2 &\leq& 2^{-nH(T) + n\d} \\
 2^{-nH(T) - n\d} &\leq& \,\,\,\Norm{\ph^{T^n}}_\infty &\leq& 2^{-nH(T) + n\d}. 
\end{IEEEeqnarray*}
The entropies in these bounds are evaluated on $\ket{\psi}^{ABC}$.  
Additionally, the normalized version $\ket{\Psi}^{A^nB^nC^n}$ of the subnormalized state
\[\big(\openone^{A^n}\ox \openone^{B^n}\ox \Pi_\delta^{C^n}\big) \big(\ket{\psi}^{ABC}\big)^{\ox n}\]
satisfies 
\begin{IEEEeqnarray*}{rCl}
\Norm{\Psi^{A^nB^nC^n} - \big(\psi^{ABC}\big)^{\ox n}}_1 &\leq& \ep.
\end{IEEEeqnarray*}
Finally, there is a $\delta$-independent constant $c>0$ such that we may take $\ep = 2^{-nc\d^2}$ in all of the above bounds. 
\label{theo:types}
\end{theo}
\vspace{.1in}

\emph{Proof of Theorem~\ref{theo:fqrspiggyback}}:  
We will apply Theorem~\ref{theo:types} two separate times to the state $\big(\ket{\psi}^{ABCD}\big)^{\ox n}$, obtaining two auxiliary states that control the main quantities in the error bound  (\ref{eqn:eta}) of Theorem~\ref{theo:1shotfqrspiggyback}.  For the first, we consider $\ket{\psi}^{ABCD}$ to be a tripartite state of the systems $A,C,BD$.  We thus obtain, for every $\d > 0$ and all sufficiently large $n$, a state $\ket{\ph}^{A^n\!B^n\!C^n\!D^n}$ that is $\ep$-close to $\ket{\psi}^{\ox n}$ in trace distance for $\ep = 2^{-nc\d^2}$, such that the matrix norms in the second term of (\ref{eqn:eta}) have the appropriate exponential bounds.  With respect to the partition $AD,B,C$,  we similarly obtain another state $\ket{\phi}^{A^n\!B^n\!C^n\!D^n}$ such that the operator norms in the last term of (\ref{eqn:eta}) are bounded accordingly.  
Alice initiates the protocol by Schumacher compressing the system $C^n$.  For this, she performs the projective measurement $\{\Pi_\delta^{C^n},\openone^{C^n} \!\!\!\!- \Pi_\delta^{C^n}\}$ on $C^n$.  According to Theorem~\ref{theo:types}, the first outcome occurs with probability at least $1-\ep$.
In this case, the global state is replaced by the normalized version $\ket{\Psi}^{A^n\!B^n\!C^n\!D^n}$ of the projected state $\Pi^{C^n}_\d\big(\ket{\psi}^{ABCD}\big)^{\ox n}$.  
In case the other outcome occurs, Alice declares an error and the protocol is aborted.
We condition on the first case.  In what follows, we identify $\ket{\Psi}^{A^n\!B^n\!C^n\!D^n}$ with its restriction $\ket{\Psi}^{A^n\!B^n\!C_\d D^n}$ to the support $C_\d$ of the typical projection $\Pi_\d^{C^n}$. 

By the triangle inequality, each of $\ket{\ph}^{A^n\!B^n\!C^n\!D^n}$ and  $\ket{\phi}^{A^n\!B^n\!C^n\!D^n}$ is $2\ep$-close to $\ket{\Psi}^{A^n\!B^n\!C_\d D^n}$ in trace distance because all three states are $\ep$-close to $\big(\ket{\psi}^{ABCD}\big)^{\ox n}$.
Therefore, 
the one-shot theorem (Theorem~\ref{theo:1shotfqrspiggyback}) implies that there exist a quantum system $S$ and a maximally entangled state $\ket{\Phi}^{\h{A}\h{B}}$ with $|\h{A}|\cdot |S| = |C_\d|$, together with $\kappa$ encoding isometries $\CV_k^{\h{A}A^n\!C_\d\to A^n\!S}$ and a decoding isometry $\CW^{SB^n\h{B}\to B^n\!C_\d K_\out}$ satisfying (\ref{eqn:1shoterror}) and (\ref{eqn:eta}) with $\ep$ replaced by $2\ep$.  If Alice applies one of the isometries $\CV_k$ uniformly at random and sends $S$ to Bob, after which he applies $\CW$, the system $C_\d$ will be transferred with high global fidelity.  By measuring $K_\out$, Bob can, on the average, identify Alice's encoding.  Rather than send Bob classical information, Alice can instead simulate a coherent channel from a system $K_\inn$ to $K_\inn K_\out$ by applying a controlled isometry 
\[\CV = \sum_k \proj{k}^{K_\inn} \ox \CV^{\h{A}A^nC_\d\to A^nS}.\]
If she tries to send half of a maximally entangled state $\ket{\Phi}^{K'K_\inn}$,
the global pure state $\ket{\Om}$ on $A^n B^nC_\d D^nK'K_\inn K_\out$ that results from the protocol is 
\[\ket{\Om} = \CW\circ\CV\ket{\Psi}^{A^n\!B^n\!C_\d D^n}\ket{\Phi}^{K'K_\inn}\ket{\Phi}^{\h{A}\h{B}}.\]
It is then immediate from (\ref{eqn:1shoterror}) that 
\[\bra{\Phi}^{K'K_\inn K_\out}\bra{\Psi}^{A^n\!B^n\!C_\d D^n}\ket{\Om} \geq 1-\eta.\]
where 
\[\ket{\Phi}^{K'K_\inn K_\out} = \frac{1}{\sqrt{\kappa}}\sum_{k=1}^\kappa 
\ket{k}^{K'}\ket{k}^{K_\inn}\ket{k}^{K_\out}\]
is a GHZ state.  
The corresponding fidelity is thus bounded by $(1-\eta)^2 \geq 1-2\eta$.
By monotonicity of fidelity, we obtain 
\begin{IEEEeqnarray}{rCl}
F(\ket{\Phi}^{K'K_\inn K_\out},\Om^{K'K_\inn K_\out}) \geq 1-2\eta
\label{eqn:est1}
\end{IEEEeqnarray}
and with (\ref{eqn:fid2tr}), we similarly find that 
\begin{IEEEeqnarray}{rCl}
\Norm{\Om^{A^n\!B^n\!C_\d D^n} - \Psi^{A^n\!B^n\!C_\d D^n}}_1 \leq 2\sqrt{2\eta}. \nn
\end{IEEEeqnarray}
Because $\Psi^{A^n\!B^n\!C_\d D^n}$ is $\ep$-close to $\big(\psi^{ABCD}\big)^{\ox n}$ 
in trace distance, the triangle inequality implies that 
\begin{IEEEeqnarray}{rCl}
\Norm{\Om^{A^n\!B^n\!C_\d D^n} - \big(\psi^{ABCD}\big)^{\ox n}}_1 \leq 2\sqrt{2\eta} + \ep.
\label{eqn:est3}
\end{IEEEeqnarray}
We may combine the estimates (\ref{eqn:est1}) and (\ref{eqn:est3}) using Lemma~2 from \cite{YDH05},  yielding
\begin{IEEEeqnarray}{rCl}
\IEEEeqnarraymulticol{3}{l}{
F\left(\ket{\Phi}^{K'K_\inn K_\out}\big(\ket{\psi}^{ABCD}\big)^{\ox n},\ket{\Om}\right) \hspace{1in}} \nn \\
&\geq& 
1 - \Norm{\Om^{A^n\!B^n\!C_\d D^n} - \big({\psi}^{ABCD}\big)^{\ox n}}_1 \nn\\
& & \hspace{.105in} - 3\big(1-F(\ket{\Phi}^{K'K_\inn K_\out},\Om^{K'K_\inn K_\out})\big)\nn \\
&\geq& 1-\ep  - 2\sqrt{2\eta}- 6\eta.  \label{eqn:finalfidelity}
\end{IEEEeqnarray}
Now we only need to bound the two main terms in the expression (\ref{eqn:eta}) for $\eta$. 
Taking
$|S| = 2^{nQ}$ and $\kappa = 2^{nR}$, 
the first main quantity in (\ref{eqn:eta}) satisfies
\begin{IEEEeqnarray}{rCl}
\IEEEeqnarraymulticol{3}{l}{
\frac{|C_\d| \Norm{\ph^{B^n\!D^n}}_0\Norm{\ph^{B^n\!C^n\!D^n}}_2^2}{|S|^2} \hspace{1.6in}} \nn\\ \hspace{1in}
&\leq& 2^{n[H(C) + H(BD) - H(BCD) - 2Q] + 3n\d} \nn\\
&=& 2^{n[I(C;BD)-2Q] + 3n\d} \label{eqn:iidbound1}
\end{IEEEeqnarray}
and thus tends to zero exponentially fast provided that 
\begin{eqnarray}
Q\geq \frac 12 I(C;BD) + 2\d. \label{eqn:Qrate}
\end{eqnarray}  For the second term, 
\begin{IEEEeqnarray*}{rCl}
 \frac{\kappa \Norm{\phi^{B^n\!C^n}}_0 \Norm{\phi^{B^n}}_\infty}{|C_\d|} &\leq&
 2^{n[R+H(BC)-H(C) -H(B)] + 3n\d} \\
&=& 2^{n[R-I(B;C)] + 3n\d} 
\end{IEEEeqnarray*}
so that if $R \leq I(B;C) - 4\d$, this term also goes to zero exponentially with $n$.   For sufficiently large $n$, each of these terms is less than $\ep = 2^{-nc\d^2}$, giving 
\[\eta \leq 6\sqrt{2\ep} + 4\ep^{1/4} + 4\ep.
\]
Therefore, the overall fidelity (\ref{eqn:finalfidelity}) is at least $1-6\ep^{1/8}$ when $\ep$ is sufficiently small. 
Recall the identity 
\[H(C) = \12 I(C;BD) + \12 I(A;C).\] 
If $Q$ obeys (\ref{eqn:Qrate}), Theorem~\ref{theo:types} implies that $|C_\delta| \leq 2^{n[H(C) + \delta]}$.  Therefore the protocol uses entanglement at rate
\begin{IEEEeqnarray*}{rCl}
E_\inn &=& \frac{1}{n}\log{|C_\d|} - Q
\leq \12 I(A;C) - \delta. 
\end{IEEEeqnarray*} 
Because $\delta>0$ can be taken arbitrarily small, it follows that whenever 
\[Q > \12 I(C;BD), \,E_\inn > \12 I(A;C),\, \text{ and } R < I(B;C),\] 
we have, for all sufficiently large $n$, 
\begin{eqnarray*}
\psi^{AC|B} + \lfloor nQ\rfloor [q\to q] + \lfloor n E_\inn\rfloor [qq]  \hspace{1.4in} \\
\hspace{1.3in} \geq_{6\ep^{1/8}}  \psi^{A | CB} +
\lfloor nR \rfloor [q\to qq].
\end{eqnarray*}
Since this holds for arbitrarily small $\ep > 0$ (in fact, it even holds for $\ep \to 0$ exponentially fast with $n$), the asymptotic resource inequality of Theorem~\ref{theo:fqrspiggyback} follows.  
\QED
\subsection{Proof of Theorem~\ref{theo:1shotfqrspiggyback} \label{sec:fqrsproof}}  
Our proof of Theorem~\ref{theo:1shotfqrspiggyback} makes essential use of the following robust one-shot decoupling lemma, which is proved in the appendix.  After stating the lemma, we briefly recall how it is used in two previously studied special cases of our redistribution result, to help the reader understand the context into which it fits with our proof.  
\begin{lem}[Robust one-shot decoupling]
Let a density matrix $\psi^{CE}$ be given, fix $\ep > 0$ and let  $\ph^{CE}$ be any state satisfying $\Norm{\psi^{CE} - \ph^{CE}}_1\leq \ep$.  
Fix a unitary decomposition $W^{C\to S\h{B}}$ of $C$ into subsystems and define,  for each $U^{C\to C}$,  
\begin{eqnarray*}
\psi_U^{S\h{B}E} &=&  WU\psi^{C E}U^\dagger W^\dagger.
\end{eqnarray*}
Then the average state 
\[\overline{\psi}^{S\h{B}E} = \int_{\CU(C)} \psi_U^{S\h{B}E} dU.\]
satisfies
\begin{IEEEeqnarray}{rCl}
\Norm{\overline{\psi}^{\h{B}E} -  \pi^{\h{B}}\ox \psi^E}_1 
&\leq & 
2 \ep + \sqrt{\frac{|C| \Norm{\ph^{E}}_0
\Norm{\ph^{CE}}_2^2}
{|S|^2}}. \hspace{.3in}
\end{IEEEeqnarray}
\label{lem:robustdecoupling}
\end{lem}
\vspace{.1in}

The state redistribution problem generalizes two previously considered tasks which nonetheless play a role in our proof.  Since our nomenclature differs from past writings, we pause briefly to describe our conventions.  When $A$ is trivial or is otherwise regarded as part of the reference $D$, we follow \cite{HOW05} in calling the corresponding task \emph{state merging} because Alice is asked to ``merge" $C$ with $B$.  When $B$ is trivial we call the task \emph{state splitting} because Alice must ``split" $C$ apart from $C$.  
In \cite{D05b,ADHW06}, these tasks were respectively called ``fully quantum Slepian-Wolf" and ``fully quantum reverse Shannon", with the additional understanding that the involved resources are quantum communication and entanglement.  On the other hand,  \cite{HOW05,HOW05b} introduced a protocol for the state merging problem --  the so-called ``state merging protocol" that only allows the use of quantum entanglement and classical communication.  In this paper, when we speak of protocols for merging and splitting, we shall mean protocols with fully quantum resources in the sense of \cite{D05b,ADHW06}, reserving the the term ``merging with classical communication" for protocols in the sense of \cite{HOW05,HOW05b}.  In Section~\ref{section:discussion}, we show how our protocol generalizes these latter protocols when the communication is limited to be only classical.

Given $\ket{\psi}^{BCD}$, an optimal protocol for merging $C$ with $B$, was given by proving the inequality 
\[\psi^{C|B} + \12 I(C;D) [q\to q] \succeq  \psi^{\emptyset|CB} + \12 I(B;C) [qq].\]
Together with the method of types (Theorem~\ref{theo:types}), the above decoupling lemma provides an immediate proof of this resource inequality. Indeed, if Alice encodes with a random unitary, Lemma~\ref{lem:robustdecoupling} ensures that a system $\h{A}$ (identified with $\h{B}$ in the lemma), which will hold Alice's half of the generated entanglement, is approximately maximally mixed and decoupled from $R$.  This can easily be shown to imply that $\h{A}$ is maximally entangled with $BS$ (see \cite{ADHW06}, or compare with the proof of Theorem~\ref{theo:1shotfqrspiggyback} below).  Because all transformations are unitary and the global state is pure, this ensures that Bob can apply a local isometry to reconstruct $C$, while at the same time obtaining the other half of the generated entanglement.  This scenario is illustrated on the left of Figure~\ref{fig:FQSWcircuit}.

\begin{figure}
\centering
\includegraphics[scale=.5]{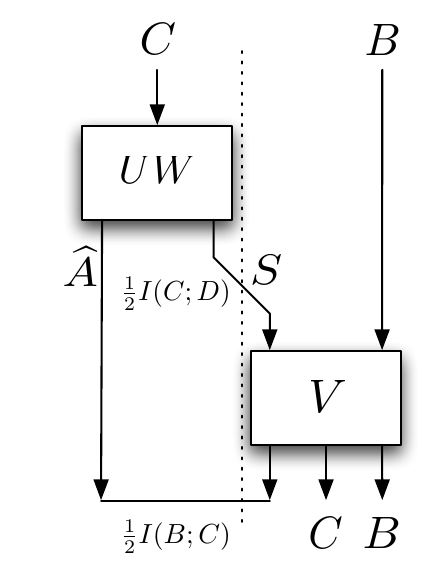} \hspace{.3in}
\includegraphics[scale=.5]{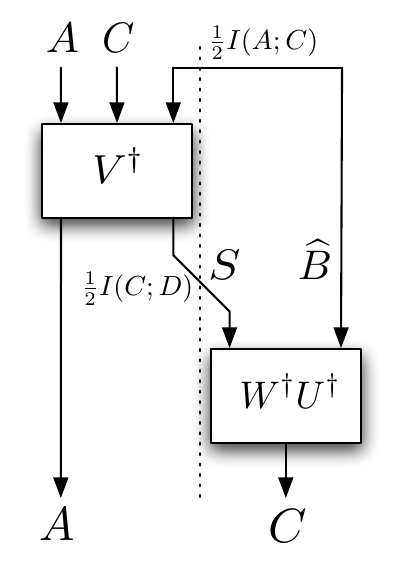}
\caption{Circuits for merging (left) and splitting (right), related by time-reversal and swapping $A\leftrightarrow B$.  We have included the rates one gets by applying the method of types (Theorem~\ref{theo:types}) to the one-shot decoupling lemma (Lemma~\ref{lem:robustdecoupling}).  Note that for merging, the random encoding determines the decoding, while for splitting, the random decoding determines the encoding.}
\label{fig:FQSWcircuit}
\end{figure}  

Given $\ket{\psi}^{ACD}$, a circuit for splitting is obtained by running a merging circuit in reverse (while swapping the labels $A \leftrightarrow B$), yielding the inequality  
\[\psi^{AC|\emptyset} + \12 I(C;D) [q\to q]  + \12 I(A;C) [qq] + \succeq \psi^{A|C}.\]
The corresponding circuit is pictured on the right of Figure~\ref{fig:FQSWcircuit}.  

We prove Theorem~\ref{theo:1shotfqrspiggyback} as follows.   If Bob's side information is considered as part of the reference (i.e.\ is disregarded as side information), the fully quantum reverse Shannon protocol can be used to transfer $C$ from Alice to Bob, at least making use of Alice's side information.  By a modification of that protocol provided below, Bob's side information can be utilized to simulate the required coherent channels $[q\to qq]$ as follows.   Rather than choosing a single random unitary for the decoding, we choose exponentially many (roughly $2^{nI(B;C)}$ when we apply the method of types to the one-shot result).  We further guarantee that if Alice chooses one of the corresponding encodings uniformly at random, Bob can, on average, correctly distinguish that encoding in order to apply the correct decoding.  Thus, it is possible for Alice to ``piggyback" classical information on the transmitted qubits, that Bob can access by means of his side information (cf.\ \cite{HHHLT01,BSST02}).  We further ensure that this can all be done coherently, where Alice instead applies a superposition of encodings by using a controlled isometry that is controlled by an arbitrary quantum state.   The circuit we construct for performing this task non-coherently is illustrated in Figure~\ref{fig:piggyback}.
\begin{figure}
\centering
\includegraphics[scale=.5]{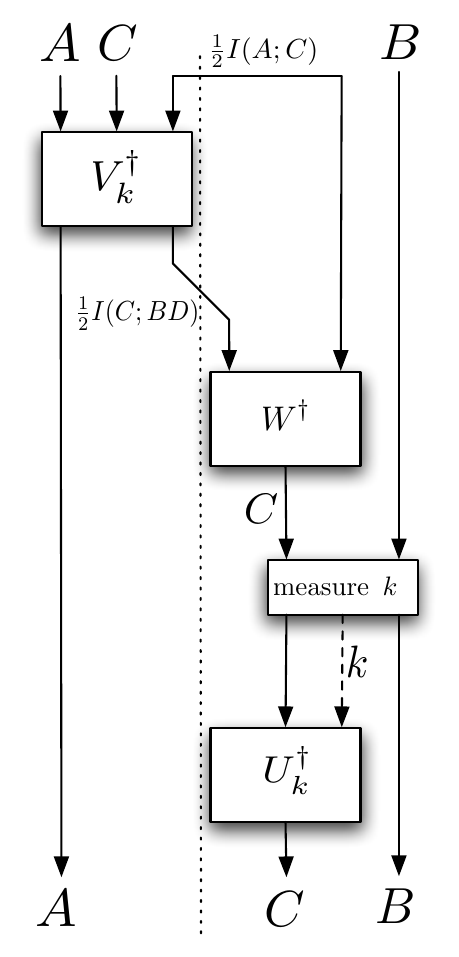}
\caption{Circuit for using Bob's side information to piggyback extra classical (or coherent) information through the fully quantum reverse Shannon circuit on the right of Figure~\ref{fig:FQSWcircuit}.}
\label{fig:piggyback}
\end{figure}  

Our proof of Theorem~\ref{theo:1shotfqrspiggyback} relies on two other lemmas.  First, we require the operator inequality  \cite{HN03}: 
\begin{lem}
If $0\leq \Pi \leq \openone$ and $\Pi \leq \Lambda$, then
\[\openone - \Lambda^{-1/2} \Pi \Lambda^{-1/2} \leq 2(\openone - \Pi) + 4(\Lambda-\Pi).\]
\label{lem:hayashi}
\end{lem}
We also will use the following coherification lemma, which allows us to convert protocols that transmit classical information to ones that simulate coherent channels.  We give a short proof in the appendix. 
\begin{lem}
Given a pure state $\ket{\psi}^{XY}$ and $\kappa$ unitaries $U_k^{X\to X}$, let $\ket{\psi_k}^{XY} = U_k\ket{\psi}^{XY}$.  Given any other set of pure states $\ket{\psi'_k}^{XY}$ and a POVM $\{\Lambda_k^X\}$ on $X$,
there are complex phases $\al_k$ such that the 
isometry 
\[\CL^{X\to XK} = \sum_k  (\al_k U^\dagger_k\sqrt{\Lambda_k}) \ox \ket{k}^K \]
satisfies 
\[\frac{1}{\kappa}\sum_{k=1}^\kappa \bra{k}\bra{\psi}\CL\ket{\psi'_k} \geq 
1 - 2(P +\sqrt{1-F}).\]
where \[P = 1-\frac 1\kappa\sum_k\Tr\psi_k\Lambda_k,\hspace{.2in}
F = \frac 1\kappa\sum_k |\braket{\psi_k}{\psi_k'}|^2.\]
\label{lem:coherent}
\end{lem}
\vspace{.1in}


\begin{proof}[Proof of Theorem~\ref{theo:1shotfqrspiggyback}]  
As in the statement of the theorem, we fix nearby states $\ket{\ph}$ and $\ket{\phi}$ and let $W^{C\to S\h{B}}$ be any unitary decomposition of $C$ into subsystems. Independently choose $\kappa$ unitaries $\{U_1,\dotsc,U_\kappa\}$ according to the Haar measure on $\CU(C)$.  
For each $k$, define the states 
\begin{eqnarray*}
\ket{\psi_k}^{ABCD} &=& U_k\ket{\psi}^{ABCD} \\
\ket{\psi_k}^{AS\h{B}BD} &=& WU_k\ket{\psi}^{ABCD} = W\ket{\psi_k}^{ABCD}. 
\end{eqnarray*}
We define
the \emph{decoupling fidelity} for $\psi^{\h{B}BD}_k$ as 
\[F_k = F(\psi_{k}^{\h{B}BD},\pi^{\h{B}}\ox\psi_k^{BD}).\]
Since $\ket{\Phi}^{\h{A}\h{B}}\ket{\psi}^{ABCD}$ 
is a purification of $\pi^{\h{B}}\ox\psi^{BD}$,  Uhlmann's theorem implies that there is an isometry $V_k^{AS \to \h{A}AC'}$ under which 
\begin{eqnarray*}
F_k=\big|\bra{\psi_k}^{AS\h{B}BD}V_k^\dagger\ket{\Phi}^{\h{A}\h{B}}\ket{\psi}^{ABCD}\big|^2.
\label{eqn:fqsw_ep}
\end{eqnarray*}
To send the message $k$, Alice will apply the isometry 
$\CV_k = V_k^\dagger$.
We now define 
\[\ket{\psi'_k}^{ABCD} =  W^\dagger \CV_k\ket{\Phi}^{\h{A}\h{B}}\ket{\psi}^{ABCD},\]
which is the state that is created after Alice performs $\CV_k$ and gives $S$ to Bob, who then applies $W^{\dagger S\h{B} \to C}$. 
We may therefore equivalently write 
\begin{equation*}
F_k=\big|\bra{\psi_k}^{ABCD} \ket{\psi_k'}^{ABCD}\big|^2. 
\end{equation*}  
 The \emph{average decoupling fidelity} is a random variable
\begin{eqnarray*}
F_\ave &=& \frac 1\kappa \sum_{k=1}^\kappa F_k
\end{eqnarray*}
that depends on the random choice of unitaries.  We lower bound its expectation 
as follows. 
Define the average states with respect to Haar measure $dU$ as
\begin{eqnarray*}
\overline{\psi}^{ABCD} &=& \int_{\CU(C')}  \psi^{ABCD}_U dU\\
\overline{\psi}^{AS\h{B}BD} &=& W\overline{\psi}^{ABCD} W^\dagger. 
\end{eqnarray*}

We now use the robust one-shot decoupling lemma (Lemma~\ref{lem:robustdecoupling}) to bound the expectation of $F_\ave$ over the random choice of unitaries:
\begin{IEEEeqnarray*}{rCl}
1 - \E F_\ave
 &=& 1 - \frac{1}{\kappa}\sum_{k=1}^\kappa \E F_k \nn \\
&=& 1 - F(\overline{\psi}^{\h{B}BD},\pi^{\h{B}}\ox\psi^{BD}) \nn \\
&\leq& 2\ep  + \sqrt{\frac{|C| \norm{\ph^{BD}}_0
\norm{\ph^{BCD}}_2^2}{|S|^2}}.
\end{IEEEeqnarray*}
A related estimate to be used later is 
\begin{IEEEeqnarray}{rCl}
\E\sqrt{1 -  F_\ave} &\leq& \sqrt{1-\E F_\ave} \nn \\
&\leq& \sqrt{2\ep} + \left( \frac{|C| \norm{\ph^{BD}}_0
\norm{\ph^{BCD}}_2^2}{|S|^2}\right)^{1/4},\hspace{.2in} 
\label{eqn:EFbound}
\end{IEEEeqnarray}
which follows by concavity and the inequality $\sqrt{x + y} \leq \sqrt{x} + \sqrt{y}$, valid for $x,y\geq 0$.

Next, we consider Bob's ability to distinguish the states $\psi'^{BC}_k$.
For this, we design a measurement that distinguishes the nearby states $\phi_{k}^{BC} =  U_k\phi^{BC}U_k^\dagger$.  
Let $\Pi$ be the projection onto the support of $\phi^{BC}$ and define 
\[\Pi_k = U_k \Pi U_k^{\dagger},\]
 while defining the ``pretty good measurement"
\begin{IEEEeqnarray*}{rClrCl}
\Lambda &=& \sum_{k=1}^\kappa \Pi_{k},\hspace{.2in} & 
\Lambda_k &=& \Lambda^{-1/2}\Pi_{k}\Lambda^{-1/2}.
\end{IEEEeqnarray*}
The probability that this measurement fails to identify the state $\psi_k'^{BC}$ is 
\[P_k = \Tr(\openone-\Lambda_k)\psi_k'^{BC}.\] 
Observe that 
\begin{IEEEeqnarray*}{rCl}
\big|P_k - \Tr(\openone - \Lambda_k)\phi_{k}^{BC}\big| 
&\leq& \Norm{\psi'^{BC}_k - \phi^{BC}_{k}}_1  \\
&\leq& \Norm{\psi'^{BC}_k - \psi^{BC}_{k}}_1 \!\!+ \!\Norm{\psi^{BC}_k - \phi^{BC}_{k}}_1  \\
&\leq& 2\sqrt{1-F_k} + \ep \\
&\equiv& D_k.
\end{IEEEeqnarray*}
Because $x\mapsto \sqrt{x}$ is concave, we have 
\[D_\ave \equiv \frac{1}{\kappa}\sum_{k=1}^\kappa D_k \leq \ep + 2\sqrt{1-F_\ave}.\]
Therefore, the average of the $P_k$ can be bounded using Lemma~\ref{lem:hayashi}, obtaining a random variable satisfying   
\begin{IEEEeqnarray*}{rCl}
P_\ave&\equiv& \frac{1}{\kappa} \sum_{k=1}^\kappa P_k \\
&\leq& D_\ave \!+\! \frac{1}{\kappa} \sum_{k=1}^\kappa \Tr(\openone-\Lambda_k)\phi_{k}^{BC} \\
 &\leq&  D_\ave\! +\! \frac{1}{\kappa}\sum_{k=1}^\kappa\Big(2\big(1-\Tr \Pi_{k}\phi^{BC}_{k}\big)\! + 4\!\sum_{k'\neq k}\!\Tr\Pi_{{k'}}\phi_{k}^{BC}\Big)\\
&=&  D_\ave \!+ \!\frac 4\kappa\sum_{k=1}^\kappa\sum_{k'\neq k}\Tr\Pi_{{k'}}\phi_{k}^{BC} 
\end{IEEEeqnarray*}
The last line holds because for each $k$, $\Pi_k$ projects onto the support of $\phi_k^{BC}$.
By taking the expectation over the random choice of unitaries, this yields 
\begin{IEEEeqnarray}{rCl}
\E  P_\ave &\leq& \E D_\ave+ 4 \kappa \E\Tr\Pi_{1}\phi^{BC}_{2} \nn\\
&=&  \E D_\ave+ 4 \kappa \Tr\big[\E\Pi_{1} \,\E\phi^{BC}_{2}\big] \nn\\
&=&  \E D_\ave+ 4 \kappa \Tr\big[\E\Pi_{1}(\pi^C\ox\phi^B_{2})\big] \nn \\
&\leq& \E D_\ave + \frac{4 \kappa \norm{\phi^{BC}}_0 \norm{\phi^B}_\infty}{|C|} \nn \\
&\leq&  2\E \sqrt{1-F_\ave} +\ep + \frac{4 \kappa \norm{\phi^{BC}}_0 \norm{\phi^B}_\infty}{|C|}.
 \label{eqn:EPbound}
\end{IEEEeqnarray}
We now apply Lemma~\ref{lem:coherent} with $X = BC$ and $Y = AD$, giving an isometry $\CL^{BC\to BCK}$ under which  
\begin{equation*}\frac 1\kappa \sum_{k=1}^\kappa \bra{k}\bra{\psi} \CL\ket{\psi'_k} \geq 1- 2(P_\ave + \sqrt{1-F_\ave}). \label{eqn:coh1}
\end{equation*}
Taking expectations, we find that 
\begin{IEEEeqnarray*}{rCl}
\IEEEeqnarraymulticol{3}{l}{
1 - \E\frac 1\kappa \sum_{k=1}^\kappa \bra{k}\bra{\psi} \CL\ket{\psi'_k}} \hspace{2in} \\
 \hspace{0in} &\leq&  2\E\sqrt{1- F_\ave}  + 2\E P_\ave\\
 &\leq& 4\E\sqrt{1-F_\ave} + \ep +  \frac{4 \kappa \norm{\phi^{BC}}_0 \norm{\phi^B}_\infty}{|C|} \\
 &\leq& 6\sqrt{\ep} +   4\left(\frac{|C| \norm{\ph^{BD}}_0
\norm{\ph^{BCD}}_2^2}{|S|^2}\right)^{\!\!1/4}
 \!\!\!\!+  \frac{4 \kappa \norm{\phi^{BC}}_0 \norm{\phi^B}_\infty}{|C|}.
\end{IEEEeqnarray*}
The second inequality is by  (\ref{eqn:EPbound}) while the third is due to (\ref{eqn:EFbound}) and holds for $\ep \leq  (6-4\sqrt{2})^2$. We may then conclude that for a particular value of the randomness, the same bound holds without the expectations.  Finally, we define Bob's decoding isometry to be $\CW = \CL W^\dagger$, completing the proof.  
\end{proof}


\section{An operational proof of strong subadditivity \label{section:ssad}}
Let $\ket{\psi}^{ABCD}$ be an arbitrary pure state.  In this section, we show how our results lead to an operational proof of strong subaddivity, i.e.\ that $I(C;D|B) \geq 0$.  
By discarding some resources on the right in Theorem~\ref{theo:aux}, we obtain:  
\begin{eqnarray*}
\psi^{AC|B} + \12 I(C;BD) [q\to q]  + \12 I(A;C)[qq]\hspace{.9in} \\  
\succeq \12 I(B;C)[q\to q]. 
\end{eqnarray*}
Intuitively, it makes sense that we should have 
\[I(C;BD) - I(B;C) = I(C;D|B) \geq 0\]
since otherwise, a noiseless qubit channel could be used to faithfully transmit more than one qubit in the presence of entanglement between the sender and receiver.  Of course this inequality is guaranteed by strong subadditivity.  However, our aim is to provide an alternative proof of this fundamental inequality. 
The above asymptotic resource inequality implies that for every $\ep, \delta > 0$ and all sufficiently large n, we have 
\begin{IEEEeqnarray}{rCl}
\Psi^{L|L'} \!\!+ \! {\Big\lfloor \! \mbox{$\frac n2$} I(C;BD) \!+ \! n\delta\!\Big\rfloor} [q\to q] 
\!\geq_\ep \! \!{\Big\lfloor\! \mbox{$\frac n2$} I(B;C)\! \Big\rfloor} [q\to q].\,\,\,\,\,\,\,\,\,\,\,
\label{eqn:ssadresource}
\end{IEEEeqnarray}
$\Psi^{L|L'}$ represents prior entanglement between Alice and Bob. Its precise form is irrelevant for our argument; we lose generality by assuming it is pure.    
Now consider the following lemma, whose proof we delay until the end of this section.
\begin{lem}
Let $X$ and $Y$ be quantum systems and let $\ket{\Psi}^{L|L'}$ be arbitrary.  Consider an attempted simulation 
\[\CN^{X\to X}\big( \rho^X \big)  = \CD^{YL'\to X}\circ \big(\CE^{XL\to Y}\ox \openone^{L'}\big)\big(\rho^X\ox \Psi^{L|L'}\big)\]
of the identity quantum channel $\id^{X\to X}$ by the possibly smaller one $\id^{Y\to Y}$, assisted by the bipartite state $\ket{\Psi}^{L|L'}$.   If $\ket{\Phi}^{X'X}$ is maximally entangled, then the \emph{entanglement fidelity} \cite{S96} satisfies 
\begin{eqnarray}
F\big(\ket{\Phi}^{X'X},(\openone^{X'}\ox \CN)(\Phi^{X'X})\big) \leq \frac{|Y|}{|X|}.
\label{eqn:entfid}
\end{eqnarray}
\label{lem:fidelitybound}
\end{lem}

Plugging in $|Y| = 2^{\lfloor \frac n2 I(C;BD) +n\delta\rfloor}$ and $|X| = 2^{\lfloor \frac n2 I(B;C)\rfloor}$ to (\ref{eqn:entfid}), we find that the entanglement fidelity is upper bounded by $2^{\lfloor \frac n2 I(C;D|B) +n(\delta + \frac 1n)\rfloor}$.  
Suppose now that strong subaddivity was not satisfied.  Then, for some sufficiently small $\delta > 0$ the entanglement fidelity would tend to zero exponentially fast with $n$.  However, (\ref{eqn:ssadresource}) implies that for sufficiently large $n$, the entanglement fidelity can be made arbitrarily close to 1.  Therefore  $I(C;D|B) \geq 0$.
\QED

\vspace{.1in}
\begin{proof}[Proof of Lemma~\ref{lem:fidelitybound}]
Let $\{E_i\}$ and $\{D_j\}$ be Kraus matrices for the encoding $\CE^{XL\to Y}$ and decoding $\CD^{YL'\to X}$. 
Fixing orthonormal bases of $L$ and $L'$ that Schmidt-decompose the assistance state as 
\begin{IEEEeqnarray*}{rCl}
\ket{\Psi}^{L|L'} &=& \sum_\ell \sqrt{\lambda_\ell}\ket{\ell}^L\ket{\ell}^{L'},
\end{IEEEeqnarray*}
the above Kraus matrices can be written in block form 
\[E_i = \Big[E_{i1} \cdots\, E_{i |L|}\Big],\,\,\,\,\,\,\,\,
D_j = \Big[D_{j1} \cdots\, D_{j |L|}\Big].\]
Because these maps are trace-preserving, we have
\[\sum_{i} E_i^\dagger E_i = \openone^{XL},\,\,\,\,\,\,\,\, \sum_j D_j^\dagger D_j = \openone^{YL}\]
which in turn implies that 
\begin{equation}
\sum_{i} E_{i\ell}^\dagger E_{i\ell'} = \delta_{\ell\ell'}\openone^{X},\,\,\,\,\,\,\,\, \sum_j D_{j\ell}^\dagger D_{j\ell'} = \delta_{\ell\ell'}\openone^{Y}.
\label{eqn:blockcondition}
\end{equation}
The overall map $\CN^{X\to X}$ has Kraus matrices  
\[N_{ij} = \sum_{\ell} \sqrt{\lambda_\ell} D_{j\ell} E_{i\ell}.\]
The entanglement fidelity (\ref{eqn:entfid}) can be written as \cite{S96}:
\begin{IEEEeqnarray*}{rCl}
F\big(\ket{\Phi}^{X'X},(\openone^{X'}\ox \CN)(\Phi^{X'X})\big) &=& \sum_{ij}\big|\!\Tr N_{ij} \pi^X\big|^2 \\
&=& \frac{1}{|X|^2}\sum_{ij}\big|\!\Tr N_{ij}\big|^2.
\end{IEEEeqnarray*}
On the other hand, 
\begin{IEEEeqnarray}{rCl}
\sum_{ij}\big|\!\Tr N_{ij}\big|^2 
&=& \sum_{\ell ij} \lambda_\ell  \big|\!\Tr D_{j\ell}E_{i\ell}\big|^2 \nn \\
&\leq& \sum_{\ell ij} \lambda_\ell |Y|\Tr E_{i\ell}^\dagger D_{j\ell}^\dagger D_{j\ell}E_{i\ell} \label{eqn:ssadderiv1}\\
&=& |Y|\sum_\ell \lambda_\ell \Tr\left[\sum_i E_{i\ell}^\dagger\Big(\sum_j D_{j\ell}^\dagger D_{j\ell}\Big)E_{i\ell}\right] \nn \\
&=& |Y|\sum_\ell\lambda_\ell \Tr \openone^X \label{eqn:ssadderiv2}\\
&=& |Y|\cdot |X|. \nn
\end{IEEEeqnarray}
Above, (\ref{eqn:ssadderiv1}) holds because for each $i,j$ and $\ell$, there is a rank $|Y|$ projection $P$ satisfying $PD_{j\ell}E_{i\ell} = D_{j\ell}E_{i\ell}$, while the Cauchy-Schwartz inequality implies
\begin{eqnarray*}
\big|\!\Tr PD_{j\ell}E_{i\ell}\big|^2  &\leq& \big(\!\Tr P^\dagger P\big)\cdot \big( \!\Tr E_{i\ell}^\dagger D_{j\ell}^\dagger D_{j\ell} E_{i\ell}\big) \\
&=& |Y| \Tr E_{i\ell}^\dagger D_{j\ell}^\dagger D_{j\ell} E_{i\ell}.
\end{eqnarray*}
Equation (\ref{eqn:ssadderiv2}) follows from the identities (\ref{eqn:blockcondition}) and the last line holds because the squares of the Schmidt coefficients sum to unity.  This proves the lemma.
\end{proof}


\section{Discussion \label{section:discussion}}

State redistribution is the most general unidirectional two-terminal fully quantum source coding problem.  It consists of moving a subsystem of a multipartite pure state between two spatially separated parties when the sender and receiver each hold subsystems, which are regarded as quantum side information.  We have identified the cost, in terms of entanglement and transmitted qubits, for performing state redistribution, by presenting a protocol that uses these two resources at optimal rates, i.e.\ that matches the Luo-Devetak outer bound \cite{LD06}.   Our proof that this protocol exists consists of a new resource inequality that, 
 when combined with other known results, implies that an optimal protocol exists.  
The optimal lower bound on the achievable communication rates provides the first known operational interpretation of quantum conditional mutual information.  
Technically, we provide an interpretation for one half of the conditional mutual information; nonetheless, we observed in \cite{DY06b} that by teleportation, we obtain a bona fide interpretation of conditional mutual information (i.e.\ without the 1/2) as the optimal communication rate when only classical communication is allowed in the sense of \cite{HOW05, HOW05b}.  While operational interpretations of quantum mutual information are known \cite{GPW05,SW06}, these do not simply lead to one for the conditional quantity by naively subtracting mutual informations.  Instead, one requires a proof consisting of a protocol (as found here) achieving rates arbitrarily close to the desired quantity, together with a converse (as in \cite{LD06}) demonstrating optimality. 

Our interpretation provides an explanation of the quadripartite pure state identity 
$I(C;D|A) = I(C;D|B)$ because the essential reversibility of our protocol implies that the communication cost is the same in both directions.  Indeed, with the exception of the Schumacher compression step, which is essentially reversible because it succeeds with high probability, the protocol constructed to prove Theorem~\ref{theo:fqrspiggyback} consists entirely of isometries.  Moreover, the additional steps used to arrive at Theorem~\ref{theo:main} introduce at most a ``sublinear amount" of nonunitarity.  Throughout this paper, we have adhered to the convention of always conditioning on Bob's side information, although this was an arbitrary notational choice.  We thus interpret quantum conditional mutual information --  as it appears throughout this paper --  as a measure of the quantum correlations between $C$ and $D$, from the perspective of \emph{either} $A$ or $B$.   

\begin{figure}
\centering
\begin{tabular}{c|cccc}
& $A$ & $B$ & $D$\\
\hline 
state redistribution & $\bullet$  & $\bullet$ & $\bullet$ \\
state merging & $\circ$  & $\bullet$ & $\bullet$ \\
state splitting & $\bullet$  & $\circ$ & $\bullet$ \\
Schumacher compression   & $\circ$  & $\circ$ & $\bullet$ \\
entanglement concentration & $\circ$  & $\bullet$ & $\circ$ \\
entanglement dilution & $\bullet$  & $\circ$ & $\circ$ \\
concentration + dilution & $\bullet$  & $\bullet$ & $\circ$ 
\end{tabular}
\caption{State redistribution reduces to other known problems when various subsystems, represented here by open circles, are trivial ($C$ is always nontrivial in these settings). We exclude the trivial problem consisting of just a pure state on $C$, which can be solved with no nonlocal resources at all.}
\label{figure:specialcases}
\end{figure}

In Figure~\ref{figure:specialcases}, we illustrate several special cases of state redistribution.  Our protocol yields optimal protocols for the problems listed there, at least with regard to the rates at which resources are consumed or generated.  
Respectively disregarding Alice's or Bob's side information  gives optimal protocols for state merging and state splitting (recall our nomenclature from Section~\ref{sec:fqrsproof}), which can also be obtained by simply  combining Theorem~\ref{theo:types}  and the robust decoupling lemma (Lemma~\ref{lem:robustdecoupling}). Furthermore, when both parties lack side information we recover (albeit somewhat trivially) Schumacher data compression.  
As pointed out  in \cite{DY06b}, the formal time-reversal duality between merging and splitting observed in \cite{D05b} is embodied in a more natural way by our new protocol, which is in fact self-dual with respect to time reversal.  In \cite{DY06b}, we also observed the intuitively satisfying -- but nonetheless surprising -- fact that successive redistribution can be performed optimally using the optimal redistribution protocol.  

Other protocols are obtained when $D$ is trivial, in which case strong subadditivity is saturated $I(C;D|B) = 0$ and thus any positive communication rate is achievable by our protocol.  When either $A$ or $B$ is also trivial, we respectively obtain protocols for entanglement concentration and dilution \cite{BBPS96}, and when both $A$ and $B$ are nontrivial, state redistribution gives an alternate approach to first concentrating the $AC|B$ entanglement then diluting the $A|CB$ entanglement \cite{HL04}, each of which gives a net entanglement cost of $H(A) - H(B)$.  Note that \cite{HW03,HL04} showed that diluting EPR entanglement into i.i.d.\ pure states requires a nonzero (but sublinear) communication cost to achieve any constant error, while exponentially small error requires any nonzero communication rate.  We therefore must expect the same with even the most generic state redistribution instances that saturate strong subadditivity. 
Here, the states are such that $C$ is conditionally decoupled from the reference $D$ given $A$ or $B$ and, 
up to local unitaries, have the form 
\[\sum_x \sqrt{p_x} \ket{x}^{A'}\ket{x}^{B'}\ket{\psi_x}^{A_CB_CC}\ket{\phi_x}^{A_DB_DD}.\]
As pointed out in \cite{DY06b} (with a sign error in the published version) this type of state can be redistributed with entanglement cost 
\[\sum_x p_x \big(H(A_C) - H(B_C)\big)_{\psi_x}.\]
An interesting problem that we do not address in this paper is to more carefully account for sublinear terms in the overall cost for redistribution.  Besides giving more precise estimates when the overall rates are zero,  a more careful study might provide a better understanding of transformations between non-maximally entangled states as considered in \cite{HL04, FL05}.  In particular, we note that while exponentially small error is generically possible with our protocol, this might not be possible when sublinear amounts of resources are used.

Because the main technical part of our proof is proved in a one-shot fashion, it could possibly be applied to more general quantum sources that do not satisfy the i.i.d.\ property but are instead structured in some other way; for instance, to ground states of many-body Hamiltonians in statistical physics.  In particular, there are intriguing connections between state redistribution and topological entanglement entropy, which is a characteristic of topologically ordered ground states of gapped 2D quantum spin systems.  These connections will be pursued elsewhere.

It could be useful for such applications to have a more direct proof of Theorem~\ref{theo:main} that does not use coherent channels or the cancellation lemma.  While it would be most desirable to have a one-shot version of Theorem~\ref{theo:main}, it might be more natural (see Note Added) to find a one-shot version of the related resource inequality 
\begin{IEEEeqnarray*}{rCl}
\psi^{AC|B} + \12 I(C;D|B) [q\to q] + \12 I(A;C)[qq] \succeq \hspace{.7in}\\
\hspace{.8in}\psi^{A|BC} + \12 I(B;C)[qq].
\end{IEEEeqnarray*}
The corresponding circuit for this case makes the time-reversal symmetry most apparent, as illustrated in Figure~\ref{fig:symmetry}.
\begin{figure}
\centering
\includegraphics[scale=.7]{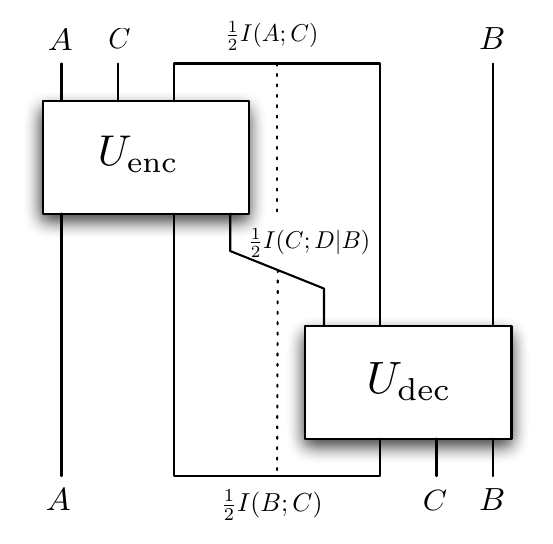}
\caption{Potential one-shot redistribution circuit making time-reversal symmetry apparent.} 
\label{fig:symmetry}
\end{figure}

We expect state redistribution to be a useful primitive for studying more complicated state transfer problems.  Most generally, one can imagine $n$ spatially separated parties all holding various parts of a global multipartite state, wishing to shuffle their subsystems around in some arbitrary but predetermined way.  There is a multitude of ways that redistribution could be applied to give achievable rate regions for such problems, where each round of communication would fit our general setting,  although they would most likely be suboptimal in general.  A simple example along these lines, for which the optimal solution is not yet known,  was considered in \cite{OW05}, where Alice and Bob wish to swap two systems.    Perhaps judicious use of state redistribution can lead to new achievable rates for this or related problems by optimizing over ways of splitting the systems to be swapped into subsystems.  

Apparently, one half of the mutual information plays a central role in characterizing the optimal rates in this paper.  In the following somewhat mysterious fashion, this quantity can be considered as a ``measure" of the correlations between two subsystems. 
By analogy with thermodynamics, it is possible to identify an underlying heuristic organizing principle governing our optimal rates that perhaps could lend itself to further generalizations of redistribution.  The main task of state redistribution is to transform between two configurations of the subsystems as follows:
\[AC\big| B \big| D \to A\big| CB \big | D.\]
Let $\CA_{\text{initial/final}}$ (resp.\ $\CB$) denote the systems Alice (resp.\ Bob) holds at the beginning/end of the protocol.  
Consider the following ``dynamic potentials" relative to Alice$\to$Bob communication:
\begin{IEEEeqnarray*}{rCl}
K_\text{initial}^{\CA\to \CB} &\equiv& \12 I(\CA_\text{initial};D) = \12 I(AC;D) \\
K_\text{final}^{\CA\to \CB}&\equiv&  \12 I(\CA_\text{final};D) 
 = \12 I(A;D).
\end{IEEEeqnarray*}
We interpret these as indicating the correlations between Alice's systems and the reference, both before and after redistribution.  
The optimal qubit rate for redistribution is easily shown to equal the difference between the dynamic potentials
\[K_\text{final}^{\CA\to \CB} - K_\text{initial}^{\CA\to \CB} = \12 I(C;D|A) = \12 I(C;D|B).\]
We are therefore operationally justified in interpreting this difference as measuring the correlations with the reference that Alice must transfer to Bob to redistribute the state.
Analogously, we may also define ``static potentials"
\begin{IEEEeqnarray*}{rCl}
S^{\CA\to \CB}_\text{initial} &\equiv& \12 I(\CA_\text{initial};\CB_\text{initial}) = \12 I(AC;B) \\
 S^{\CA\to \CB}_\text{final} &\equiv&  \12 I(\CA_\text{final};\CB_\text{final}) = \12 I(A;BC)
 \end{IEEEeqnarray*}
that indicate the correlations between Alice's and Bob's systems at each state of redistribution. Similarly, the optimal ebit rate can be shown to equal the difference of the static potentials
\[S^{\CA\to \CB}_\text{final} -  S^{\CA\to \CB}_\text{initial} = \12 I(A;C) - \12 I(B;C).\]
It is operationally justifiable to consider this difference as the amount of excess correlation between Alice and Bob that is involved in going between the two configurations.  

Relative to the Bob$\to$Alice direction, the dynamic potentials are subtracted from a constant
\[K^{\CB\to \CA}_{\text{initial/final}} = H(D) - K^{\CA\to \CB}_{\text{initial/final}}\]
while the static potentials obey 
\[S^{\CA\to \CB}_{\text{initial/final}} = S^{\CB\to \CA}_{\text{final/initial}}.\] 
Subtracting these potentials as above, we find that 
\[K^{\CB\to \CA}_{\text{final}} - K^{\CB\to \CA}_{\text{initial}} = K^{\CA\to\CB}_{\text{final}} - K^{\CA\to \CB}_{\text{initial}} ,\]
while 
\[S^{\CA\to \CB}_\text{final} -  S^{\CA\to \CB}_\text{initial} = -\big(S^{\CB\to \CA}_\text{final} -  S^{\CB\to \CA}_\text{initial}\big),\]
providing another explanation of the symmetry properties of the optimal rates.  
One could imagine generalizations of the above where more complicated potentials are defined for redistribution problems involving many more parties.  However, we expect it would be challenging to find operational justifications for such theories.

\section*{Acknowledgments}
We would like to thank Charlie Bennett for suggesting the circuit pictured in Figure~\ref{fig:symmetry} and Toby Berger for encouraging us to find a quantum counterpart to the classical result on successive refinement of information.  Igor Devetak was supported in part by the NSF grants CCF-0524811 and CCF-0545845 (CAREER).  Jon Yard's research at Caltech was supported from the NSF under the grant PHY-0456720.   His research at LANL is supported by the Center for Nonlinear Studies (CNLS), the Quantum Institute and the LDRD program of the U.S.\ Department of Energy.

\section*{Note added}
After a preprint of this article was made available, a one-shot version of our main result along the lines of Figure~\ref{fig:symmetry} was found \cite{Opp08,YBW08}.


\appendix
Here we collect the proofs of some auxiliary results used in the proof of Theorem~\ref{theo:1shotfqrspiggyback}.  Our proof of the robust decoupling lemma (Lemma~\ref{lem:robustdecoupling}) relies on the following non-robust version from \cite{ADHW06}.
\begin{lem}[One-shot decoupling]
Let a density matrix $\ph^{CE}$ be given and fix a unitary decomposition $W^{C\to S\h{B}}$ of $C$ into subsystems.  For each unitary $U^{C\to C}$, define 
\[\ph_U^{S\h{B}E} =  WU\ph^{CE}U^\dagger W^\dagger.\]
Then 
\begin{IEEEeqnarray}{rCl}
\int_{\CU(C)} \Norm{\ph^{\h{B}E}_U - \pi^{\h{B}}\ox \ph^E}^2_1dU 
&\leq& \frac{|C| \Norm{\ph^E}_0 
\Norm{\ph^{CE}}_2^2}{|S|^2}. \hspace{.3in}
\end{IEEEeqnarray}
\label{lem:decoupling}
\end{lem}

\begin{proof}[Proof of Lemma~\ref{lem:robustdecoupling}]
By convexity of the trace norm
\[\Norm{\overline{\psi}^{\h{B}E} - \pi^{\h{B}}\ox \psi^E}_1
\leq \int_{\CU(C)} \Norm{\psi^{\h{B}E}_{U} - \pi^{\h{B}}\ox \psi^E}_1dU,\]
where $dU$ is Haar measure on $\CU(C)$.
We use the triangle inequality to bound the integrand:
\begin{IEEEeqnarray}{rCl}
 \Norm{\psi^{\h{B}E}_{U} - \pi^{\h{B}}\ox \psi^E}_1 
 &\leq&  \Norm{\ph^{\h{B}E}_{U} - \pi^{\h{B}}\ox \ph^E}_1 \\
& +&  \Norm{\psi^{\h{B} E}_U -\ph^{\h{B}E}_U}_1 \\
& +&  \Norm{\pi^{\h{B}}\ox \psi^{E} - \pi^{\h{B}}\ox \ph^E}_1.
\end{IEEEeqnarray}
The second term is bounded using monotonicity, unitary invariance of the trace norm, and the assumed $\ep$-closeness of $\psi^{CE}$ and $\ph^{CE}$:  
\begin{IEEEeqnarray}{rCl}
\Norm{\psi^{\h{B} E}_U -\ph^{\h{B}E}_U}_1 
&\leq& \Norm{\psi^{S\h{B} E}_U -\psi^{S\h{B}E}_U}_1 \nn\\
&=& \Norm{\psi^{CE} -\ph^{CE}}_1 \nn\\
&\leq& \ep.
\end{IEEEeqnarray}
Similarly, the last term satisfies  
\begin{IEEEeqnarray*}{rCl}
\Norm{\pi^{\h{B}}\ox \psi^{E} - \pi^{\h{B}}\ox \ph^E}_1 
&=& \Norm{\psi^{E} -  \ph^E}_1 \\
&\leq& \Norm{\psi^{CE} -\ph^{CE}}_1 \nn\\
&\leq& \ep.
\end{IEEEeqnarray*}
Because $x\mapsto x^2$ is convex, the integral of the first term satisfies
\begin{eqnarray*}
\left(\int_{\CU(C)} \Norm{\ph^{\h{B}E}_{U} - \pi^{\h{B}}\ox \ph^E}_1dU\right)^2 \hspace{1in}\\
\hspace{1in} \,\,\leq \,\,\int_{\CU(C)} \Norm{\ph^{\h{B}E}_{U} - \pi^{\h{B}}\ox \ph^E}^2_1dU. & 
\end{eqnarray*} 
The theorem follows by applying Lemma~\ref{lem:decoupling} to this integral. 
\end{proof}

\begin{proof}[Proof of Lemma~\ref{lem:coherent}]
To begin, note that we may choose the complex phases so that $\bra{k}\bra{\psi}\CL\ket{\psi'_k} = |\bra{k}\bra{\psi}\CL\ket{\psi'_k}|$. Now 
\begin{IEEEeqnarray*}{rCl}
\bra{k}\bra{\psi}\CL\ket{\psi'_k} &\geq& \big(\bra{k}\bra{\psi}\CL\ket{\psi'_k}\big)^2 \\
&\geq&  |\bra{k}\bra{\psi}\CL\ket{\psi_k}|^2 - \norm{\CL(\psi_k) - \CL(\psi_k')}_1.
\end{IEEEeqnarray*}
Because $0\leq \Lambda_k \leq \openone$, we have  
\begin{IEEEeqnarray*}{rCl}
 |\bra{k}\bra{\psi}\CL\ket{\psi_k}|^2 &=& |\bra{\psi_k}\sqrt{\Lambda_k}\ket{\psi_k}|^2 \\
&\geq&  |\bra{\psi_k}\Lambda_k\ket{\psi_k}|^2 \\
&=& \big(\!\Tr\psi_k\Lambda_k\big)^2. 
\end{IEEEeqnarray*}
Furthermore, unitary invariance of the trace norm and (\ref{eqn:fid2tr}) imply that 
\[\norm{\CL(\psi_k) - \CL(\psi_k')}_1 = \norm{\psi_k - \psi_k'}_1 \leq 2\sqrt{1 - |\braket{\psi_k}{\psi'_k}|^2}.\]
Therefore, 
\begin{IEEEeqnarray*}{rCl}
\bra{k}\bra{\psi}\CL\ket{\psi'_k} \geq \big(\!\Tr\psi_k\Lambda_k\big)^2  - 2\sqrt{1 - |\braket{\psi_k}{\psi'_k}|^2}.
\end{IEEEeqnarray*}
Finally, because the functions $x\mapsto x^2$ and $x\mapsto -\sqrt{x}$ are convex, we find that 
\begin{IEEEeqnarray*}{rCl}
\frac{1}{\kappa}\sum_{k=1}^\kappa \bra{k}\bra{\psi}\CL\ket{\psi'_k} &\geq&
(1-P)^2 - 2\sqrt{1-F} \\ 
&\geq& 1-2(P + \sqrt{1-F}) 
\end{IEEEeqnarray*}
as required.
\end{proof}


\bibliographystyle{IEEEtran}

\end{document}